\documentclass[12pt, journal, draftclsnofoot, onecolumn]{IEEEtran}

\usepackage[utf8]{inputenc} 
\usepackage[T1]{fontenc}
\usepackage{url}
\usepackage{ifthen}
\usepackage[cmex10]{amsmath}
\IEEEoverridecommandlockouts

\usepackage{color}
\usepackage{graphicx,tabularx,array,amsmath,amsthm,thmtools}

\usepackage{mathtools}

\usepackage{amsfonts}
\usepackage{bm}
\usepackage{bbm}
\usepackage{makecell}
\usepackage{multirow}
 \usepackage{amssymb}
\usepackage{txfonts}
\usepackage[T1]{fontenc}
\usepackage{tikz}
\usepackage[scr=dutchcal]{mathalfa}
\usepackage{ textcomp }

\usepackage{cite}
\newcommand\blfootnote[1]{%
  \begingroup
  \renewcommand\thefootnote{}\footnote{#1}%
  \addtocounter{footnote}{-1}%
  \endgroup
}
\usepackage{amsthm}
\theoremstyle{definition}
\newtheorem{Theorem}{Theorem}
\newtheorem{Proposition}{Proposition}
\newtheorem{Lemma}{Lemma}

\newtheorem{Example}{Example}
\newtheorem{Remark}{Remark}
\newtheorem{Definition}{Definition}

\theoremstyle{definition} 
\begin{document}

\title{On the Joint Typicality of Permutations of Sequences of Random Variables
}

\author{Farhad Shirani, Siddharth Garg, and Elza Erkip \\
Electrical and Computer Engineering Department \\
New York University, NY, USA \\\date{} }
\maketitle
\begin{abstract}
Permutations of correlated sequences of random variables appear naturally in a variety of applications such as graph matching and asynchronous communications. In this paper, the asymptotic statistical behavior of such permuted sequences is studied. It is assumed that a collection of random vectors is produced based on an arbitrary joint distribution, and the vectors undergo a permutation operation. The joint typicality of the resulting permuted vectors with respect to the original distribution is investigated. 
As an initial step, permutations of pairs of correlated random vectors are considered.  
    It is shown that the probability of joint typicality of the permuted vectors depends only on the number and length of the disjoint cycles of the permutation. Consequently, it suffices to study typicality for a class of permutations called \textit{standard permutations}, for which,  upper-bounds on the probability of joint typicality are derived. The notion of standard permutations is extended to  a class of permutation vectors called \textit{Bell permutation vectors}. 
By investigating Bell permutation vectors, upper-bounds on the probability of joint typicality of permutations of arbitrary collections of random sequences are derived.  
\blfootnote{This work is supported by NYU WIRELESS Industrial Affiliates and
National Science Foundation grant CCF-1815821.} 
\end{abstract}

\section{Introduction}
The notion of weak typicality was first introduced by Shannon \cite{Shannon} in studying the data compression problem. Later, Wolfowitz \cite{wolfowitz2012coding} introduced strong typicality to provide alternative proofs for Shannon's channel coding theorem in data transmission. In the past several decades, typicality has become one of the essential components in the information theoretic study of point-to-point and multiterminal communication systems  \cite{el2011network}. 
Typicality is also useful in other applications such as graph matching \cite{shirani2018matching, ped}, database matching \cite{shiranidatabase}, and group testing \cite{atia2012boolean}, where sequences of noisy observations of an original dataset are observed. In these instances, the relationship between the original data and the observed data can be captured through a \textit{`test channel'} which induces the noise on the observations, and  joint typicality  can be used to identify correct matches between the original and observed data. 

The conventional approach in information theory which uses typicality to study communication systems relies heavily on the assumption of synchronous communication. To elaborate, let us consider the transmission of data over a discrete memoryless channel. In this problem, it is assumed that an $n$-length vector $X^n$ drawn in an independendently and identically distributed (i.i.d.) fashion from distribution $P_X$ is input to a channel characterized by the transition probability $P_{Y|X}$ over $n$ uses of the channel, and the output vector $Y^n$ is produced. The receiver may then leverage the fact that with high probability, the pair $(X^n,Y^n)$ is jointly typical with respect to the joint distribution $P_{X}P_{Y|X}$ to recover the transmitted message. 

In other problems of interest such as graph matching \cite{shirani2018matching} and database matching \cite{shirani2019concentration}, the \textit{receiver} in the \textit{test channel} does not know the order of the received signals. For instance, the graph matching problem considers a pair of randomly generated graphs with $n$ vertices and correlated $n\times n$ adjacency matrices $(\mathbf{X}, \mathbf{Y})$. It is assumed that we are given the adjacency matrix $\mathbf{X}$, and the permuted  adjacency matrix $\pi(\mathbf{Y})$, where the permutation is due to a relabeling of the underlying graph. The objective is to recover the permutation $\pi$ by leveraging the correlation among $\mathbf{X}$ and $\mathbf{Y}$. 
This necessitates the study of the probability of joint typicality of pairs of vectors of the form $(\mathbf{X},\pi(\mathbf{Y}))$. 

In this paper, we investigate the typicality of permutations of sequences of correlated random vectors.
We first consider pairs of independently and i.i.d. random vectors $(X^n,Y^n)$ generated according to a joint distribution $P_{XY}$. For a given permutation $\pi$ of $n$-length sequences, we are interested in finding the probability of joint typicality of $(X^n,\pi(Y^n))$ with respect to the distribution $P_{XY}$. We first show that the probability of joint typicality only depends on the number and length of the disjoint cycles of the permutation $\pi$. Consequently, we define a class of permutation vectors called \textit{standard permutations}, such that for any given number and length of the disjoint cycles, there is exactly one unique standard permutation. We derive bounds on the probability of joint typicality of $(X^n,\pi(Y^n))$ with respect to the distribution $P_{XY}$, for any standard permutation $\pi$. Furthermore, we consider typicality of collections of permuted sequences $(\pi_j(X^n_{(j)}), j\in \{1,2,\cdots,m\}$, where $\pi_j$ and $X^n_{(j)}$ are the $j$th permutation and $j$th sequence, respectively. We extend the notion of standard permutations to collections of more than two permuted sequences and introduce the class of \textit{Bell permutation vectors}, and derive bounds on the probability of joint typicality of $(\pi_j(X^n_{(j)}), j\in \{1,2,\cdots,m\}$ for any given Bell permutation vector $
\pi$. 

The rest of the paper is organized as follows:  Section \ref{sec:prilim} provides the necessary background on permutations. Section \ref{sec:pair} studies joint typicality of pairs of permuted sequences of random variables. Section \ref{sec:coll} extends the results to collections of permuted random vectors. Section \ref{sec:con} concludes the paper.





\textit{Notation:} 
 Random variables are represented by capital letters such as $X, U$ and their realizations by small letters such as $x, u$.
 Sets are denoted by calligraphic letters such as $\mathcal{X}, \mathcal{U}$. The probability of the event $\mathcal{A}\subset\mathcal{X}$ is denoted by $P_X(\mathcal{A})$, and the subscript is omitted when there is no ambiguity.
 The set of natural numbers, and real numbers are shown by $\mathbb{N}$, and $\mathbb{R}$ respectively. The random variable $\mathbbm{1}_{\mathcal{E}}$ is the indicator function of the event $\mathcal{E}$.
 The set of numbers $\{n,n+1,\cdots, m\}, n,m\in \mathbb{N}$ is represented by $[n,m]$. Furthermore, for the interval $[1,m]$, we sometimes use the shorthand notation $[m]$ for brevity. 
 For a given $n\in \mathbb{N}$, the $n$-length vector $(x_1,x_2,\hdots, x_n)$ is written as $x^n$.

\section{Preliminaries}
\label{sec:prilim}
We follow the notation used in \cite{isaacs} in our study of permutation groups which is summarized below. 
\begin{Definition}[\bf Permutation]
\label{def:perm1}
A permutation on the set $[1,n], n\in \mathbb{N}$ is a bijection $\pi:[1,n]\to [1,n]$. The set of all permutations on the set $[1,n]$ is denoted by $\mathcal{S}_n$.
\end{Definition}
\begin{Definition}[\bf Cycle and Fixed Point]
 A permutation $\pi \in \mathcal{S}_n, n\in \mathbb{N}$  is called a cycle if there exists $k\in [1,n]$ and $\alpha_1,\alpha_2,\cdots,\alpha_k\in [1,n]$ such that i) $\pi(\alpha_i)=\alpha_{i+1}, i\in [1,k-1]$, ii) $\pi(\alpha_n)=\alpha_1$, and iii) $\pi(\beta)=\beta$ if $\beta\neq \alpha_i, \forall i\in [1,k]$. The variable $k$ is the length of the cycle. The element $\alpha$ is a fixed point of the permutation if $\pi(\alpha)=\alpha$. We write $\pi=(\alpha_1,\alpha_2,\cdots,\alpha_k)$.
 The cycle $\pi$ is non-trivial if $k\geq 2$. 
\end{Definition}

\begin{Lemma}[\hspace{-.005in}\bf\cite{isaacs}]
 Every permutation $\pi \in \mathcal{S}_n, n\in \mathbb{N}$ has a unique decomposition into disjoint non-trivial cycles.  
\end{Lemma}

\begin{Definition}
For a given $n,m,c\in \mathbb{N}$, and $1\leq i_1\leq i_2\leq \cdots\leq i_c \leq n$ such that $n=\sum_{j=1}^ci_j+m$, an $(m,c,i_1,i_2,\cdots,i_c)$-permutation is a permutation in $\mathcal{S}_n$ which has $m$ fixed points and $c$ disjoint cycles with lengths $i_1,i_2,\cdots,i_c$, respectively.
\label{def:cycle}
\end{Definition}

\begin{Example}
Consider the permutation which maps the vector $(1,2,3,4,5)$ to $(5,1,4,3,2)$. The permutation can be written as a decomposition of disjoint cycles in the following way $\pi=(1,2,5)(3,4)$, where $(1,2,5)$ and $(3,4)$ are cycles with lengths $3$ and $2$, respectively. The permutation $\pi$ is a $(0,2,2,3)$-permutation.
\end{Example}

\begin{Definition}[\bf Sequence Permutation]
\label{def:perm2}
 For a given sequence $y^n\in \mathbb{R}^n$ and permutation $\pi\in \mathcal{S}_n$, the sequence $z^n=\pi(y^n)$ is defined as $z^n=(y_{\pi(i)})_{i\in [1,n]}$.\footnote{Note that in Definitions \ref{def:perm1} and \ref{def:perm2} we have used  $\pi$ to denote both a scalar function which operates on the set $[1,n]$ as well as a function which operates on the vector space $\mathbb{R}^n$.}
\end{Definition}

\begin{Definition}[\bf Derangement]
A permutation on vectors of length $n\in \mathbb{N}$ is called a derangement if it does not have any fixed points. The number of distinct derangements of $n$-length vectors is denoted by $!n$. 
\end{Definition}

\section{Permutations of Pairs of Sequences}
As a  first step, we consider typicality of permutations of pairs of correlated sequences. 
\begin{Definition}[\bf Strong Typicality \cite{cover}]
\label{Def:typ}
Let the pair of random variables $(X,Y)$ be defined on the probability space $(\mathcal{X}\times\mathcal{Y},P_{XY})$, where $\mathcal{X}$ and $\mathcal{Y}$ are finite alphabets. The $\epsilon$-typical set of sequences of length $n$ with respect to $P_{XY}$ is defined as:
\begin{align*}
&\mathcal{A}_{\epsilon}^n(X,Y)=\Big\{(x^n,y^n): \Big|\qquad\frac{1}{n}N(\alpha,\beta|x^n,y^n)-P_{XY}(\alpha,\beta)\Big|\leq \epsilon, \forall (\alpha,\beta)\in \mathcal{X}\times\mathcal{Y}\Big\},
\end{align*}
where $\epsilon>0$, $n\in \mathbb{N}$, and $N(\alpha,\beta|x^n,y^n)= \sum_{i=1}^n \mathbbm{1}\left((x_i,y_i)=(\alpha,\beta)\right)$.  
\end{Definition}

For a correlated pair of independent and identically distributed (i.i.d) sequences $(X^n,Y^n)$ and arbitrary permutations $\pi_x,\pi_y\in \mathcal{S}_n$, we are interested in bounding the probability $P((\pi_x(X^n),\pi_y(Y^n))\in \mathcal{A}_{\epsilon}^n(X,Y))$.

\label{sec:pair}
In our analysis, we make extensive use of the standard permutations defined below.
\begin{Definition}[\bf Standard Permutation]
Let $m,c,i_1,i_2,\cdots,i_c$ be as in Definition \ref{def:cycle}.
The $(m,c,i_1,i_2,\cdots,i_c)$-standard permutation is defined as the $(m,c,i_1,i_2,\cdots,i_c)$-permutation consisting of the cycles $(\sum_{j=1}^{k-1}i_j+1,\sum_{j=1}^{k-1}i_j+2,\cdots,\sum_{j=1}^{k}i_j), k\in [1,c]$. Alternatively, the $(m,c,i_1,$ $i_2,\cdots,i_c)$-standard permutation is defined as:
\begin{align*}
&\pi=(1,2,\cdots,i_1)(i_1+1,i_1+2,\cdots,i_1+i_2)\cdots
\\&(\sum_{j=1}^{c-1}i_j+1,\sum_{j=1}^{c-1}i_j+2,\cdots,\sum_{j=1}^{c}i_j)(n-m+1)(n-m+2)\cdots (n).
\end{align*}
\label{def:stan_perm}
\end{Definition}
\begin{Example}
The $(2,2,3,2)$-standard permutation is a permutation which has $m=2$ fixed points and $c=2$ cycles. The first cycle has length $i_1=3$ and the second cycle has length $i_2=2$. It is a permutation 
on sequences of length $n=\sum_{j=1}^ci_j+m=3+2+2=7$. The permutation is given by $\pi= (1 2 3)(4 5)(6)(7)$. For an arbitrary sequence $\alpha^7=(\alpha_1,\alpha_2,\cdots,\alpha_7)$, we have:
\begin{align*}
 \pi(\alpha^7)=(\alpha_3,\alpha_1,\alpha_2,\alpha_5,\alpha_4,\alpha_6,\alpha_7).
\end{align*}
\end{Example}

The following proposition shows that in order to find bounds on the probability of joint typicality of permutations of correlated sequences, it suffices to study pairs of permuted sequences $(X^n,\pi(Y^n))$, where $\pi$ is an standard permutation. 

\begin{Proposition}
\label{Prop:1}
 Let $(X^n,Y^n)$ be a pair of i.i.d sequences defined on finite alphabets. We have:
\\ i) For an arbitrary permutation $\pi\in \mathcal{S}_n$, 
 \begin{align*}
 P((\pi(X^n),\pi(Y^n))\in \mathcal{A}_{\epsilon}^n(X,Y))=P((X^n,Y^n)\in \mathcal{A}_{\epsilon}^n(X,Y)).
\end{align*}
ii)  Following the notation in Definition \ref{def:stan_perm}, let $\pi_1$ be an arbitrary $(m,c,i_1,i_2,\cdots,i_c)$-permutation  and let $\pi_2$ be the $(m,c,i_1,i_2,\cdots,i_c)$-standard permutation. Then, 
\begin{align*}
 P((X^n,\pi_1(Y^n))\in \mathcal{A}_{\epsilon}^n(X,Y))=P((X^n,\pi_2(Y^n))\in \mathcal{A}_{\epsilon}^n(X,Y)).
 \end{align*}
 iii) For arbitrary permutations $\pi_x,\pi_y\in \mathcal{S}_n$, there let $\pi$ be the standard permutation having the same number of cycles and cycle lengths as that of $\pi_x^{-1}(\pi_y)$. Then,
 we have:
\begin{align*}
 P((\pi_x(X^n),\pi_y(Y^n))\in \mathcal{A}_{\epsilon}^n(X,Y))=P((X^n,\pi(Y^n))\in \mathcal{A}_{\epsilon}^n(X,Y)).
\end{align*}
\end{Proposition}
\begin{proof}
Appendix \ref{Ap:Prop:1}.
\end{proof}

The following theorem provides an upper-bound on the probability of joint typicality of a permutation of correlated sequences for a permutation with $m\in [n]$ fixed points.  
\begin{Theorem}
Let $(X^n,Y^n)$ be a pair of i.i.d sequences defined on finite alphabets $\mathcal{X}$ and $\mathcal{Y}$, respectively. For any permutation $\pi$ with $m\in [n]$ fixed points, the following holds:
 \begin{align}
  \label{eq:perm_bound}
  &P((X^n,\pi(Y^n))\in \mathcal{A}_{\epsilon}^n(X,Y)) 
  \\&\leq 2^{-\frac{n}{4}(D(P_{XY}
 ||(1-\alpha)P_XP_Y+ \alpha P_{XY})-|\mathcal{X}||\mathcal{Y}|\epsilon+O(\frac{\log{n}}{n}))}, \nonumber
\end{align}
where $\alpha= \frac{m}{n}$, and $D(\cdot||\cdot)$ is the Kullback-Leibler divergence. 
\label{th:1}
\end{Theorem}
\begin{proof}
Appendix \ref{Ap:th:1}.
\end{proof}
\begin{Remark}
The upper-bound in Equation \eqref{eq:perm_bound} goes to $0$ as $n\to \infty$ for any non-trivial permutation (i.e. $\alpha$ bounded away from one) and small enough $\epsilon$, as long as $X$ and $Y$ are not independent. 
\end{Remark}

The exponent $D(P_{XY}
 ||(1-\alpha)P_XP_Y+ \alpha P_{XY})$ in Equation \eqref{eq:perm_bound} can be interpreted as follows: for the fixed points of the permutation ($\alpha$ fraction of indices), we have $\pi({Y}_i)=Y_i$. As a result, the joint distribution of the elements  $(X_i,\pi({Y}_i))$ is $P_{XY}$. For the rest of the elements, $\pi({Y}_i)$ are permuted components of $Y^n$, as a result $(X_i,\pi({Y}_i))$ are an independent pair of variables since $(X^n,Y^n)$ is a correlated pair of i.i.d. sequences. Consequently, the distribution of $(X_i,\pi({Y}_i))$ is $P_XP_Y$ for $(1-\alpha)$ fraction of elements which are not fixed points of the permutation. The average distribution is $(1-\alpha)P_{X}P_{Y}+\alpha P_{XY}$ which appears as the second argument in the Kullback-Leibler Divergence in Equation \eqref{eq:perm_bound}. 

Theorem \ref{th:1} provides bounds on the probability of joint typicality of $X^n$ and $\pi(Y^n)$ as a function of the number of fixed points $m$ of the permutation $\pi(\cdot)$. Such bounds are often used in error analysis and derivation of error bounds in various applications \cite{tuncel2005error,shirani2018typicality,csiszar1998method}. The standard method in such analysis is to use a union bounding technique to break the error event into a set of components each pertaining to the  joint typicality of a pair of vectors $(X^n,\pi(Y^n))$. Then, an upper-bound on the probability of error is derived by counting the number of terms $(X^n,\pi(Y^n))$ for which $P((X^n,\pi(Y^n))\in \mathcal{A}_{\epsilon}^n(X,Y))$ is equal to each other and multiplying the total number of terms by that probability. 
From Theorem \ref{th:1}, for permutations of pairs of random vectors $P((X^n,\pi(Y^n))\in \mathcal{A}_{\epsilon}^n(X,Y))$ is `almost' the same for all permutations with equal number of fixed points. 
As a result, in evaluating error exponents a parameter of interest is the number of distinct permutations with a specific number of fixed points and its limiting behavior.

\begin{Lemma}
\label{lem:dercount}
Let $n\in \mathbb{N}$. Let $N_{m}$ be the number of distinct permutations with exactly $m\in [0,n]$ fixed points. Then, 
\begin{align}
    \frac{n!}{m!(n-m)}\leq  N_m= {n \choose m} !(n-m)\leq n^{n-m}.
    \label{eq:der1}
\end{align}
Particularly, let $m= \alpha n, 0<\alpha<1$. Then, the following holds:
\begin{align}
    \lim_{n\to \infty} \frac{\log{N_m}}{n\log{n}}= 1-\alpha.
    \label{eq:der2}
\end{align}
\end{Lemma}
\begin{proof}
Appendix \ref{App:lem:dercount}.
\end{proof}

In the following, we investigate whether the exponent in Equation \eqref{eq:perm_bound} is tight (i.e. whether the exponent can be improved to arrive at a tighter upper-bound). Previously, we provided the justification for the appearance of the term $D(P_{XY}
 ||(1-\alpha)P_XP_Y+ \alpha P_{XY})$ in the exponent in Equation \eqref{eq:perm_bound}. However, a more careful analysis may yield improvements in the coefficient $\frac{n}{4}$ by focusing on specific classes of permutations as described in the following. As a first step, we only consider permutations consisting of a single non-trivial cycle and no fixed points. 
 
 \begin{Lemma}
 \label{Lem:No_fixed}
Let $(X^n,Y^n)$ be a pair of i.i.d sequences defined on finite alphabets $\mathcal{X}$ and $\mathcal{Y}$, respectively. For any permutation $\pi$ with no fixed points, and a single cycle (i.e. $m=0$ and $c=1$), the following holds:
 \begin{align}
  &P((X^n,\pi(Y^n))\in \mathcal{A}_{\epsilon}^n(X,Y))  
  \leq 2^{-\frac{n}{2}(I(X;Y)-\delta)},
\end{align}
where $\delta= 2\sum_{x,y}|\log_2{\frac{P_{XY}(x,y)}{P_X(x)P_Y(y)}}|\epsilon$ and $\epsilon>0$.
 \end{Lemma}
 \begin{proof}
Appendix \ref{Ap:Lem:No_fixed}. 
 \end{proof}

The following lemma derives similar results for permutations with a large number of cycles lengths bounded from above by a constant $s<n$.  
 
\begin{Lemma}
\label{Lem:short}
Let $(X^n,Y^n)$ be a pair of correlated sequences of i.i.d variables defined on finite alphabets $\mathcal{X}$ and $\mathcal{Y}$, respectively. For any  $(m,c,i_1,i_2,\cdots,i_c)$-permutation $\pi$ with no fixed points (m=0), where $0< i_1<i_2<\cdots<i_c<s<n$, the following holds:
 \begin{align}
  &P((X^n,\pi(Y^n))\in \mathcal{A}_{\epsilon}^n(X,Y))  
  \leq 2^{-\frac{n}{s}(I(X;Y)-\delta)},
\end{align}
where $\delta= \sum_{x,y}|\log_2{\frac{P_{XY}(x,y)}{P_X(x)P_Y(y)}}|\epsilon$ and $\epsilon>0$.
\end{Lemma} 
\begin{proof}
Appendix \ref{Ap:Lem:short}.
\end{proof}

\begin{Remark}
 Note that Theorem \ref{th:1} can also be applied to derive a bound on the probability of joint typicality given the permutation considered in Lemma \ref{Lem:No_fixed}. In this case $\alpha=\frac{m}{n}=0$ and $D(P_{XY}||\alpha P_{XY}+(1-\alpha)P_XP_Y)= I(X;Y)$ and Theorem 1 yeilds the exponent $\frac{n}{4}I(X;Y)$ for the probability of joint typicality. Hence, Lemma \ref{Lem:No_fixed} improves the exponent $\frac{n}{4}I(X;Y)$ in Theorem \ref{th:1} to $\frac{n}{2}I(X;Y)$ for single-cycle permutations with no fixed points. Similarly, Lemma \ref{Lem:short} improves the exponent in Theorem \ref{th:1} when the maximum cycle length is less than or equal to $s=3$. 
\end{Remark}

\section{Typicality of Permutations of Collections of Correlated Sequences}
\label{sec:coll}
In the next step, we consider joint typicality of permutations of more than two correlated sequences $(X^n_{(1)},X^n_{(2)},\cdots,X^n_{(k)}), n\in \mathbb{N}, k>2$. 

\begin{Definition}[\bf Strong Typicality of Collections of Sequences]
\label{Def:ctyp}
Let the random vector $X^k$ be defined on the probability space $(\prod_{j\in [k]}\mathcal{X}_j,P_{X^k})$, where $\mathcal{X}_j, j\in [k]$ are finite alphabets, and $k>2$. The $\epsilon$-typical set of sequences of length $n$ with respect to $P_{X^k}$ is defined as:
\begin{align*}
&\mathcal{A}_{\epsilon}^n(X^k)=\Big\{(x_{(j)}^n)_{j\in [k]}: \Big|
\\&\qquad \frac{1}{n}N(\alpha^k|x^n_{(1)},x^n_{(2)},\cdots,x^n_{(k)})-P_{X^k}(\alpha^k)\Big|\leq \epsilon, \forall \alpha^k\in \prod_{j\in [k]}\mathcal{X}_j\Big\},
\end{align*}
where $\epsilon>0$, $(x_{(j)}^n)_{j\in [k]}=(x_{(1)}^n, \cdots, x^n_{(k)})$ is a vector of sequences,  and $N(\alpha^k|x^n_{(1)},$ $x^n_{(2)}$ $,\cdots,x^n_{(k)})= \sum_{i=1}^n \mathbbm{1}\big((x_{(j),i})_{j\in [k]}=$ $\alpha^k\big)$.  
\end{Definition} 
In the previous section, in order to investigate the typicality of permutations of pairs of correlated sequences, we introduced standard permutations which are completely characterized by the number of fixed points, number of cycles, and cycle lengths of the permutation. The concept of standard permutations does not extend naturally when there are more than two sequences (i.e. more than one non-trivial permutation). Consequently, investigating typicality of permutations of collections of sequences requires developing additional analytical tools which are described in the following. 

\begin{Definition}[\bf Bell Number \cite{comtet2012advanced}]
 Let $\mathsf{P}=\{\mathcal{P}_1,\mathcal{P}_2,\cdots, \mathcal{P}_{b_k}\}$ be the set of all partitions of $[1,k]$.  The natural number $b_k$ is the $k$'th Bell number. 
\end{Definition}
In the following, we define Bell permutation vectors which are analogous to standard permutations for the case when the problem involves more than one non-trivial permutation.

\begin{Definition}[\bf Partition Correspondence]
\label{Def:corr}
Let $k,n\in \mathbb{N}$ and $(\pi_1,\pi_2,\cdots,\pi_k)$ be arbitrary permutations operating on $n$-length vectors.
The index $i\in [1,n]$ is said to correspond to the partition $\mathcal{P}_j\in \mathsf{P}$ of the set $[1,k]$ if the following holds: 
\begin{align*}
    \forall l,l'\in [1,k]: \pi^{-1}_l(i)=\pi^{-1}_{l'}(i) \iff \exists r: l,l' \in \mathcal{D}_{j,r},
\end{align*}
where $\mathcal{P}_j=\{\mathcal{D}_{j,1},\mathcal{D}_{j,2},\cdots,\mathcal{D}_{j,|\mathcal{P}_j|}\}$.
\end{Definition}

\begin{Example}
Let us consider a triple of permutations of $n$-length sequences, i.e. $k=3$, and the partition $\mathcal{P}= \{\{1,2\},\{3\}\}$. Then an index $i\in [n]$ corresponds to the partition $\mathcal{P}$ if the first two permutations map the index to the same integer and the third permutation maps the index to a different integer. 
\end{Example} 

\begin{Definition}[\bf Bell Permutation Vector]
Let $(i_1,i_2,\cdots,i_{b_k})$ be an arbitrary sequence, where $\sum_{k\in [b_k]}{i_k}=n, i_k\in [0,n]$, $b_k$ is the $k$th Bell number, and $n,k \in \mathbb{N}$.
 The vector of permutations $(\pi_1,\pi_2,\cdots, \pi_k)$ is called an $(i_1,i_2,\cdots,i_{b_k})$-Bell permutation vector if for every partition $\mathcal{P}_k$ exactly $i_k$ indices correspond to that partition. Equivalently:
\begin{align*}
    &\forall j\in [b_k]: i_k=|\{i\in [n]: \forall l,l'\in [k]: \pi^{-1}_l(i)=\pi^{-1}_{l'}(i) \\&\iff \exists r: l,l' \in \mathcal{D}_{j,r}\}|,
\end{align*}
where $\mathcal{P}_j=\{\mathcal{D}_{j,1},\mathcal{D}_{j,2},\cdots,\mathcal{D}_{j,|\mathcal{P}_j|}\}$
\label{def:bell}.
\end{Definition}

The definition of Bell permutation vectors is further clarified through the following example. 

\begin{Example}
Consider 3 permutations $(\pi_1,\pi_2,\pi_3)$ of vectors with length 7, i.e. $k=3$ and $n=7$. Then, $b_k=5$ and we have:
\begin{align*}
    &\mathcal{P}_1=\{\{1\},\{2\},\{3\}\}, 
    \quad 
    \mathcal{P}_2=\{\{1,2\},\{3\}\},
    \quad 
    \mathcal{P}_3=\{\{1,3\},\{2\}\},
    \\
    &\qquad\qquad\qquad \mathcal{P}_4=\{\{1\},\{2,3\}\}, 
    \quad 
    \mathcal{P}_5=\{\{1,2,3\}\}.
\end{align*}
Let $\pi_1$ be the trivial permutation fixing all indices and let $\pi_2= (1 3 5)(2 4)$, $\pi_3= (1 5)(2 4)(3 7)$. Then:
\begin{align*}
    & \pi_1((1,2,\cdots,7))=(1,2,3,4,5,6,7),\\ 
    & \pi_2((1,2,\cdots,7))=(5,4,1,2,3,6,7),
    \\
     & \pi_3((1,2,\cdots,7))=(5,4,7,2,1,6,3),
\end{align*}
 Then, the vector $(\pi_1,\pi_2,\pi_3)$ is a $(2,1,0,3,1)$-Bell permutation vector, where the indices $(3,5)$ correspond to the $\mathcal{P}_1$ partition (each of the three permutations map the index to a different integer), index $7$ corresponds to the $\mathcal{P}_2$ partition (the first two permutations map the index to the same integer which is different from the one for the third permutation), indices $(1,2,4)$  correspond to the $\mathcal{P}_3$ permutation (the second and third permutations map the index to the same integer which is different from the output of the first permutation),
and index $6$ corresponds to $\mathcal{P}_5$ (all permutations map the index to the same integer).
\end{Example}

\begin{Remark}
Bell permutation vectors are not unique. In other words, there can be several distinct $(i_1,i_2,\cdots, i_{b_k})$-Bell permutation vectors for given $n,k,i_1,i_2,\cdots,i_{b_k}$.
This is in contrast with standard permutations defined in Definition \ref{def:stan_perm}, which are unique given the parameters $n,k,c,i_1,i_2,\cdots,i_c$. 
\end{Remark}

The following bounds the probability of joint typicality of permutations of collections of correlated sequences:

\begin{Theorem}
\label{th:2}
Let $(X_{(j)}^n)_{j\in [k]}$ be a collection of correlated sequences of i.i.d random variables defined on finite alphabets $\mathcal{X}_{(j)}, j\in [k]$. For any $(i_1,i_2,\cdots,i_{b_k})$-Bell permutation vector $(\pi_1,\pi_2,\cdots,\pi_k)$, the following holds:
 \begin{align}
  \label{eq:cperm_bound}
  &P((\pi_i(X_{(i)} ^n)_{i\in [k]}\in \mathcal{A}_{\epsilon}^n(X^k))
  \nonumber
  \\&\leq 2^{-\frac{n}{k(k-1)b_k}(D(P_{X^k}
 ||\sum_{j\in [b_k]}\frac{i_j}{n}
 P_{X_{\mathcal{P}_j}})-\epsilon\prod_{j\in [k]}|\mathcal{X}_j|+O(\frac{\log{n}}{n}))},
\end{align}
where $P_{X_{\mathcal{P}_{j}}}=\prod_{l\in [1,|\mathcal{P}_j|]}P_{X_{j_1},X_{j_2},\cdots,X_{j_{|\mathcal{D}_{j,r}|}}}$, $\mathcal{D}_{j,r}=\{l_1,l_2,\cdots, l_{|\mathcal{D}_{j,r}|}\}, j\in [b_k], r\in [1,|\mathcal{P}_j|]$, and $D(\cdot||\cdot)$ is the Kullback-Leibler divergence. 
\label{th:cperm}
\end{Theorem}
\begin{proof}
Appendix \ref{Ap:th:2}.
\end{proof}
Note that for permutations of pairs of sequences of random variables, $k=2$ and the second Bell number is $b_2=2$. In this case $k(k-1)b_k=4$, and the bound on the probability of joint typicality given in Theorem \ref{th:cperm} recovers the one in Theorem \ref{th:1}. 

Building up on Lemma \ref{lem:dercount}, in the following, we provide upper and lower bounds on the number of distinct Bell permutation vectors for a given vector $(i_1,i_2,\cdots,i_{b_k})$. Such upper bounds may be used in evaluating error exponents as mentioned in Section \ref{sec:pair}.

\begin{Definition}[\bf k-fold Derangement]
A vector $(\pi_1(\cdot),\pi_2(\cdot),\cdots, \pi_{k}(\cdot))$ of permutations of $n$-length sequences is called an r-fold derangement if  $\pi_1(\cdot)$ is the identity permutation, and $\pi_l(i)\neq \pi_{l'}(i), l,l'\in [k], l\neq l', i\in [n]$. The number of distinct r-fold derangements of $[n]$ is denoted by $d_k(n)$. Particularly $d_{2}(n)=!n$ is the number of derangements of $[n]$.  
\end{Definition}

\begin{Lemma}
\label{lem:foldcount}
Let $n\in \mathbb{N}$ and $k\in [n]$. Then,
\begin{align*}
  ((n-k+1)!)^{k-1}  \leq d_r(n)\leq (!n)^{k-1}.
\end{align*}
\end{Lemma}
\begin{proof}
Appendix \ref{Ap:lem:foldcount}.
\end{proof}

\begin{Lemma}
\label{lem:bell}
Let $(i_1,i_2,\cdots, i_{b_k})$ be a vector of non-negative integers such that $\sum_{j\in [b_k]}i_j=n$. Define $N_{i_1,i_2,\cdots,i_{b_k}}$ as the number of distinct $(i_1,i_2,\cdots, i_{b_k})$-Bell permutation vectors. Then, 
\begin{align}
   &{n \choose i_1,i_2,\cdots, i_{b_k}}\prod_{j\in [b_k]}d_{|\mathcal{P}_{j}|}(i_j)
   \leq N_{i_1,i_2,\cdots,i_{b_k}} \\&\qquad\leq {n \choose i_1,i_2,\cdots, i_{b_k}} n^{\sum_{j\in [b_k]} |\mathcal{P}_j|i_j-n}.
    \label{eq:bell1}
\end{align}
Particularly, let $i_k=\alpha_k \cdot n, n\in \mathbb{N}$. 
The following holds:
\begin{align}
    \lim_{n\to\infty} \frac{\log{N_{i_1,i_2,\cdots,i_{b_k}}}}{n\log{n}}= \sum_{j\in [b_k]}{|\mathcal{P}_j|}{\alpha_j}-1. 
    \label{eq:bell2}
\end{align}
\end{Lemma}
\begin{proof}
Appendix \ref{Ap:lem:bell}.
\end{proof}
\section{Conclusion}
\label{sec:con}
In this paper, we have investigated the joint typicality of permutations of sequences of random vectors. As an initial step, we have considered the probability of joint typicality for pairs of permuted vectors. We have shown that this probability depends only on the number and length of the disjoint cycles of the permutation. Consequently, we have shown that it suffices to focus on a specific class of permutations called standard permutations. We have further extended the analysis to probability of joint typicality of collections of random vectors.
 \begin{appendices}
 \section{Proof of Proposition \ref{Prop:1}}
 \label{Ap:Prop:1}
 The proof of part i) follows from the fact that permuting both $X^n$ and $Y^n$ by the same permutation does not change their joint type. For part ii), it is known that there exists a permutation $\pi$ such that $\pi(\pi_1)=\pi_2(\pi)$ \cite{isaacs}. Then the statement is proved using part i) as follows: 
  \begin{align*}
& P\left(\left(X^n,\pi_1\left(Y^n\right)\right)\in \mathcal{A}_{\epsilon}^n\left(X,Y\right)\right)
= P\left(\left(\pi\left(X^n\right),\pi\left(\pi_1\left(Y^n\right)\right)\right)\in \mathcal{A}_{\epsilon}^n\left(X,Y\right)\right)
\\&= P\left(\left(\pi\left(X^n\right),\pi_2\left(\pi\left(Y^n\right)\right)\right)\in \mathcal{A}_{\epsilon}^n\left(X,Y\right)\right)
\stackrel{(a)}{=} P\left(\left(\widetilde{X}^n,\pi_2\left(\widetilde{Y}^n\right)\right)\in \mathcal{A}_{\epsilon}^n\left(X,Y\right)\right)
\stackrel{(b)}{=} P\left(\left(X^n,\pi_2\left(Y^n\right)\right)\in \mathcal{A}_{\epsilon}^n\left(X,Y\right)\right),
 \end{align*}
 where in (a) we have defined $(\widetilde{X}^n,\widetilde{Y}^n)=(\pi(X^n),\pi(Y^n))$. and (b) holds since $(\widetilde{X}^n,\widetilde{Y}^n)$ has the same distribution as $(X^n,Y^n)$.

 \section{Proof of Theorem \ref{th:1}}
 \label{Ap:th:1}
 Define the following partition for the set of indices $[1,n]$:

\begin{align*}
&\mathcal{B}_{0}= \{1, i_1+1, i_1+i_2+1, \cdots, \sum_{j=1}^{r-1} i_j +1\},
\\&
\mathcal{B}_1= \{k| \text{k is even}, k\notin \mathcal{B}_0,  k\leq \sum_{i=1}^r i_j\},\\
&\mathcal{B}_2= \{k| \text{k is odd}, k\notin \mathcal{B}_0, k\leq \sum_{i=1}^r i_j\},
\\& \mathcal{B}_3= \{k| k>\sum_{i=1}^r i_j\}.
\end{align*}
The set $\mathcal{B}_1$ is the set of indices at the start of each cycle in $\pi$, the sets $\mathcal{B}_2$ and $\mathcal{B}_3$ are the sets of odd and even indices which are not start of any cycles and $\mathcal{B}_4$ is the set of fixed points of $\pi$. Let $Z^n=\pi(Y^n)$. It is straightforward to verify that $(X_i,Z_i), i\in \mathcal{B}_j, j\in [3] $ are three sequences of independent and identically distributed variables which are distributed according to $P_{X}P_{Y}$.  The reason is that the standard permutation shifts elements of a sequence by at most one position, whereas the elements in the sequences $(X_i,Z_i), i\in \mathcal{B}_j, j\in [3] $ are at least two indices apart and are hence independent of each other (i.e. $Z_i\neq Y_i)$. Furthermore, $(X_i,Z_i), i\in \mathcal{B}_4$ is a sequence of independent and identically distributed variables which are distributed according to $P_{X,Y}$ since $Z_i=Y_i$. Let $\underline{T}_j, j\in [4]$ be the type of the sequence $(X_i,Z_i), i\in \mathcal{B}_j, j\in [4] $, so that $T_{j,x,y}= \frac{\sum_{i\in \mathcal{B}_j}\mathbbm{1}(X_i=x, Z_i=y)}{|\mathcal{B}_j|}, j,x,y\in [4]\times\mathcal{X}\times  \mathcal{Y}$. We are interested in the probability of the event $(X^n,Z^n)\in \mathcal{B}_{\epsilon}^n(X,Y)$. From Definition \ref{Def:typ} this event can be rewritten as follows:
\begin{align*}
 &P\left(\left(X^n,Z^n\right)\in \mathcal{A}_{\epsilon}^n(X,Y)\right)
  =P\left(\underline{T}(X^n,Y^n)\stackrel{.}{=} P_{X,Y}(\cdot,\cdot)\pm \epsilon\right)
 \\&= P(\alpha_1\underline{T}_1+\alpha_2\underline{T}_2+\alpha_3\underline{T}_3+\alpha_4\underline{T}_4\stackrel{.}{=}P_{X,Y}(\cdot,\cdot)\pm \epsilon),
\end{align*}
where $\alpha_i= \frac{|\mathcal{B}_i|}{n}, i\in [4]$, we write $a\stackrel{.}{=}x\pm \epsilon$ to denote $x-\epsilon \leq a\leq x_\epsilon$, and addition is defined element-wise. We have: 
\begin{align*}
 &P((X^n,Z^n)\in \mathcal{B}_{\epsilon}^n(X,Y))
 =\sum_{(\underline{s}_1,\underline{s}_2,\underline{s}_3,\underline{s}_4)\in \mathcal{T}} P(\underline{T}_i=\underline{s}_i, i\in [4]),
\end{align*}
where $\mathcal{T}= \{(\underline{s}_1,\underline{s}_2,\underline{s}_3,\underline{s}_4):\alpha_1\underline{s}_1+\alpha_2\underline{s}_2+\alpha_3\underline{s}_3+\alpha_4\underline{s}_4\stackrel{.}{=}n(P_{X,Y}(\cdot,\cdot)\pm \epsilon)\}$. Using the property that for any set of events, the probability of the intersection is less than or equal to the geometric average of the individual probabilities, we have:
\begin{align*}
 &P((X^n,Z^n)\in \mathcal{A}_{\epsilon}^n(X,Y))
\leq \sum_{(\underline{s}_1,\underline{s}_2,\underline{s}_3,\underline{s}_4)\in \mathcal{T}} \sqrt[4]{\Pi_{i\in [4]}P(\underline{T}_i=\underline{s}_i)}.
\end{align*}
Since the elements  $(X_i,Z_i), i\in \mathcal{B}_j, j\in [4] $ are i.i.d, it follows from standard information theoretic arguments \cite{csiszarbook} that:

\begin{align*}
    & P(\underline{T}_i=\underline{s}_i) \leq 2^{-|\mathcal{B}_i|(D(\underline{s}_i||P_XP_Y)-|\mathcal{X}||\mathcal{Y}|\epsilon)}, i\in [3],
    \quad P(\underline{T}_4=\underline{s}_4) \leq 2^{-|\mathcal{B}_4|(D(\underline{s}_4||P_{X,Y})-|\mathcal{X}||\mathcal{Y}|\epsilon)}.
\end{align*}
We have, 
\begin{align*}
 &P((X^n,Z^n)\in \mathcal{A}_{\epsilon}^n(X,Y))
 \\&\leq \!\!\!\!\!\!\!\! \sum_{(\underline{s}_1,\underline{s}_2,\underline{s}_3,\underline{s}_4)\in \mathcal{T}}\!\!\!\!\!\!\!\! \sqrt[4]{2^{-n(\alpha_1D(\underline{s}_1||P_XP_Y)+\alpha_2D(\underline{s}_2||P_XP_Y)+\alpha_3D(\underline{s}_3||P_XP_Y)+\alpha_4D(\underline{s}_4||P_{X,Y})-|\mathcal{X}||\mathcal{Y}|\epsilon)}}\\
 &\stackrel{(a)}{\leq}
 \sum_{(\underline{s}_1,\underline{s}_2,\underline{s}_3,\underline{s}_4)\in \mathcal{T}} \sqrt[4]{2^{-n(D(\alpha_1\underline{s}_1+\alpha_2\underline{s}_2+\alpha_3\underline{s}_3+\alpha_4\underline{s}_4
 ||(\alpha_1+\alpha_2+\alpha_3)P_XP_Y+ \alpha_4P_{X,Y})-|\mathcal{X}||\mathcal{Y}|\epsilon)}}
 \\&
 =
 |\mathcal{T}| \sqrt[4]{2^{-n(D(P_{X,Y}
 ||(1-\alpha)P_XP_Y+ \alpha P_{X,Y})-|\mathcal{X}||\mathcal{Y}|\epsilon)}}
 \stackrel{(b)}{\leq} 2^{-\frac{n}{4}(D(P_{X,Y}
 ||(1-\alpha)P_XP_Y+ \alpha P_{X,Y})-|\mathcal{X}||\mathcal{Y}|\epsilon+O(\frac{\log{n}}{n}))},
\end{align*}
where the (a) follows from the convexity of the divergence function and (b) follows by the fact that the number of joint types grows polynomially in $n$.

\section{Proof of Lemma \ref{lem:dercount}}
\label{App:lem:dercount}
First, we prove Equation \eqref{eq:der1}. Note that 
\begin{align*}
    N_m= {n \choose m} !(n-m)\leq  {n \choose m} (n-m)!= \frac{n!}{m!}\leq n^{n-m}.
\end{align*}
This proves the right hand side of the equation. To prove the left hand side, we first argue that the iterative inequality $!n\geq !(n-1)(n-1)$ holds. In other words, the number of derangements of numbers in the interval $[n]$ is at least $n-1$ times the number of derangements of the numbers in the interval $[n-1]$. We prove the statement by constructing $!(n-1)(n-1)$ distinct derangements of the numbers $[n]$. Note that a derangement $\pi(\cdot)$ of $[n]$ is characterized by the vector $(\pi(1),\pi(2),\cdots(n))$. There are a total of $n-1$ choices for $\pi(1)$ (every integer in $[n]$ except for $1$). Once $\pi(1)$ is fixed, the rest of the vector $(\pi(2),\pi(3),\cdots, \pi(n))$ can be constructed using any derangement of the set of numbers $[n]- \{\pi(1)\}$. There are a total of $!(n-1)$ such derangements. So, we have constructed $(n-1)!(n-1)$ distinct derangements of $[n]$. Consequently. $!n\geq !(n-1)(n-1)$. By induction, we have $!n\geq (n-1)!$. So,
\begin{align*}
    N_m= {n \choose m} !(n-m)\geq {n \choose m} (n-m-1)! = \frac{n!}{m!(n-m)}. 
\end{align*}
Next, we prove that Equation \eqref{eq:der2} holds. Note that from the right hand side of Equaation \eqref{eq:der1} we have:
\begin{align*}
    \lim_{n\to \infty}\frac{\log{N_m}}{n\log{n}}\leq \lim_{n\to \infty} \frac{\log{n^{n-m}}}{n\log{n}}= 
    \lim_{n\to \infty}\frac{n-m}{n}=1-\alpha.
\end{align*}
Also, from the left hand side of Equation \eqref{eq:der2}, we have: 
\begin{align*}
     &\lim_{n\to \infty} \frac{\log{N_m}}{n\log{n}}
     \geq \lim_{n\to \infty}\frac{\log{ 
     \frac{n!}{m!(n-m)}}}{n\log{n}}
      = 
      \lim_{n\to \infty}\frac{\log{ {\frac{n!}{m!}}}}{n\log{n}}-
     \frac{{{\log{ (n-m)}}}}{n\log{n}}.
\end{align*}
The second term in the last inequality converges to 0 as $n\to \infty$. Hence, 
\begin{align*}
    &\lim_{n\to \infty} \frac{\log{N_m}}{n\log{n}}
    \geq 
     \lim_{n\to \infty}\frac{\log{ {\frac{n!}{m!}}}}{n\log{n}}
     \\& \stackrel{(a)}{\geq} 
      \lim_{n\to \infty}\frac{\log{ {\frac{n!}{m^m}}}}{n\log{n}}
      {\geq} \lim_{n\to\infty}
     \frac{\log{{n!}}}{n\log{n}}
     -\frac{\log{{m^m}}}{n\log{n}}
     \stackrel{(b)}{\geq} \lim_{n\to\infty}
     \frac{n\log{n}-n+O(\log{n})}{n\log{n}}
     -\frac{\log{{m^m}}}{n\log{n}}
     \\ &
     =\lim_{n\to\infty}
     \frac{n\log{n}}{n\log{n}}
     -\frac{\alpha n\log{{\alpha n}}}{n\log{n}}
     = 1-\alpha,
\end{align*}
where in (a) we have used the fact that $m!\leq m^m$, and (b) follows from Stirling's approximation. This completes the proof.
 \section{Proof of Lemma \ref{Lem:No_fixed}}
 \label{Ap:Lem:No_fixed}
 The proof builds upon some of the techniques developed in \cite{chen}. Let $\mathcal{A}=\{(x,y)\in \mathcal{X}\times\mathcal{Y}\big| P_XP_Y(x,y)<P_{X,Y}(x,y)\}$. Let $Z^{\{(x,y)\}}_{(\pi),i}=\mathbbm{1}(X_i,Y_{\pi(i)}=(x,y))$. We have: 
\begin{align*}
&P((X^n,\pi(Y^n))\in \mathcal{A}_{\epsilon}^n(X,Y))\leq
\\& P\Big(\Big(\bigcap_{(x,y)\in \mathcal{A}}\big\{\frac{1}{n}\sum_{i=1}^nZ^{\{(x,y)\}}_{(\pi),i}>P_{X,Y}(x,y)-{\epsilon}\big\}\Big)
 \bigcap
 \Big(\bigcap_{(x,y)\in \mathcal{A}^c}\big\{\frac{1}{n}\sum_{i=1}^nZ^{\{(x,y)\}}_{(\pi),i}<P_{X,Y}(x,y)+{\epsilon}\big\}\Big) 
  \Big)
\end{align*}
For brevity let $\alpha_{x,y}=\frac{1}{n}\sum_{i=1}^nZ^{\{(x,y)\}}_{(\pi),i}$,
and $t_{x,y}=\frac{1}{2}\log_e{\frac{P_{X,Y}(x,y)}{P_X(x)P_Y(y)}}, x,y\in \mathcal{X}$. 
  Then, 
\begin{align*}
   &Pr\Big(\Big(\bigcap_{(x,y)\in \mathcal{A}}\big\{n\alpha_{x,y}>nP_{X,Y}(x,y)-n{\epsilon}\big\}\Big)
 \bigcap
\Big(\bigcap_{(x,y)\in \mathcal{A}^c}\big\{n\alpha_{x,y}<nP_{X,Y}(x,y)+n{\epsilon}\big\}\Big) 
  \Big)\\
&=  Pr\Big(\bigcap_{(x,y)\in \mathcal{X}\times\mathcal{Y}}\big\{e^{nt_{x,y}\alpha_{x,y}}>e^{nt_{x,y}P_{X,Y}(x,y)+n{\epsilon_{x,y}}}\big\}\Big),
\end{align*}
where $\epsilon_{x,y}= t_{x,y}(1-2\mathbbm{1}(x,y\in \mathcal{A}))\epsilon$ and  we have used the fact that by construction:
\begin{align}
\begin{cases}
t_{x,y}>0 \qquad& \text{if} \qquad (x,y)\in \mathcal{A}\\
t_{x,y}<0 &\text{if} \qquad (x,y)\in \mathcal{A}^c.
\end{cases} 
\label{eq:cases}
\end{align}
 So,
\begin{align}
 \nonumber  &P\Big(\Big(\bigcap_{(x,y)\in \mathcal{A}}\big\{n\alpha_{x,y}>nP_{X,Y}(x,y)-n{\epsilon}\big\}\Big)
 \bigcap
 \Big(\bigcap_{(x,y)\in \mathcal{A}^c}\big\{n\alpha_{x,y}<nP_{X,Y}(x,y)+n{\epsilon}\big\}\Big) 
  \Big)\\
&\stackrel{(a)}{\leq}
P\Big(\prod_{(x,y)\in \mathcal{X}\times\mathcal{Y}}e^{nt_{x,y}\alpha_{x,y}}>\prod_{(x,y)\in \mathcal{X}\times\mathcal{Y}}e^{nt_{x,y}P_{X,Y}(x,y)-n{\epsilon_{x,y}}}\Big)\\
&\stackrel{(b)}{\leq} e^{-\sum_{x,y}n(t_{x,y}P_{X,Y}(x,y)-{\epsilon}_{x,y})}\mathbb{E}(\prod_{x,y}e^{nt_{x,y}\alpha_{x,y}})
= 
e^{-\sum_{x,y}n(t_{x,y}P_{X,Y}(x,y)-{\epsilon}_{x,y})}\mathbb{E}(e^{\sum_{i=1}^{{n}}\sum_{x,y}t_{x,y}Z^{\{(x,y)\}}_{(\pi),i}})
\\&\stackrel{(c)}{\leq} e^{-\sum_{x,y}n(t_{x,y}P_{X,Y}(x,y)-{\epsilon}_{x,y})}\mathbb{E}^{\frac{1}{2}}(e^{\sum_{i\in \mathcal{O}}\sum_{x,y}2t_{x,y}Z^{\{(x,y)\}}_{(\pi),i}})
\mathbb{E}^{\frac{1}{2}}(e^{\sum_{i\in \mathcal{E}}\sum_{x,y}2t_{x,y}Z^{\{(x,y)\}}_{(\pi),i}})
\\&
= e^{-\sum_{x,y}n(t_{x,y}P_{X,Y}(x,y)-{\epsilon}_{x,y})}\prod_{i\in \mathcal{O}}\mathbb{E}^{\frac{1}{2}}(e^{\sum_{x,y}2t_{x,y}Z^{\{(x,y)\}}_{(\pi),i}})
\prod_{i\in \mathcal{E}}\mathbb{E}^{\frac{1}{2}}(e^{\sum_{x,y}2t_{x,y}Z^{\{(x,y)\}}_{(\pi),i}})
,
\label{eq:temp3}
\end{align}
where $\mathcal{O}$ and $\mathcal{E}$ are the odd and even indices in the set $[1,n]$. In (a) we have used the fact that the exponential function is increasing and positive, (b) follows from the Markov inequality and (c) follows from the Cauchy-Schwarz inequality. Note that: 
\begin{align*}
    \mathbb{E}(e^{\sum_{x,y}2t_{x,y}Z^{\{(x,y)\}}_{(\pi),i}})
    \stackrel{(a)}{=}\sum_{x,y} P_{X}(x)P_{Y}(y) e^{2t_{x,y}}= \sum_{x,y} P_{X}(x)P_{Y}(y) e^{\log_e{\frac{P_{X,Y}(x,y)}{P_X(x)P_Y(y)}}}
    =\sum_{x,y} P_{X,Y}(x,y)=1,
\end{align*}
where in (a) we have used the fact that $X_i$ and $Y_{\pi(i)}$ are independent since the permutation does not have any fixed points. Consequently, we have shown that: 
\begin{align*}
    Pr((X^n,\pi(Y^n))\in \mathcal{A}_{\epsilon}^n(X,Y))
    &\leq
    e^{-\sum_{x,y}n(t_{x,y}P_{X,Y}(x,y)-{\epsilon}_{x,y})}
    = e^{-\sum_{x,y}n(\frac{1}{2}P_{X,Y}(x,y)\log_e{\frac{P_{X,Y}(x,y)}{P_X(x)P_Y(y)}}-{\epsilon}_{x,y})}
    = 2^{-\frac{1}{2}n(I(X;Y)-\delta)}.
\end{align*}
This completes the proof.
 \section{Proof of Lemma \ref{Lem:short}}
 \label{Ap:Lem:short}
 The proof follows by similar arguments as that of Lemma \ref{Lem:No_fixed}. Following similar steps, we have
\begin{align}
\nonumber
&P((X^n,\pi(Y^n))\in \mathcal{A}_{\epsilon}^n(X,Y))=
 Pr\Big(\bigcap_{(x,y)\in \mathcal{X}\times\mathcal{Y}}\big\{e^{\frac{n}{s}t_{x,y}\alpha_{x,y}}>e^{\frac{n}{s}t_{x,y}P_{X,Y}(x,y)+n{\epsilon_{x,y}}}\big\}\Big)
 \\\nonumber&
 {\leq}
P\Big(\prod_{(x,y)\in \mathcal{X}\times\mathcal{Y}}e^{\frac{n}{s}t_{x,y}\alpha_{x,y}}>\prod_{(x,y)\in \mathcal{X}\times\mathcal{Y}}e^{\frac{n}{s}t_{x,y}P_{X,Y}(x,y)-\frac{n}{s}{\epsilon_{x,y}}}\Big)
\\\nonumber& 
\leq e^{-\sum_{x,y}\frac{n}{s}(t_{x,y}P_{X,Y}(x,y)-{\epsilon}_{x,y})}\mathbb{E}(\prod_{x,y}e^{\frac{n}{s}t_{x,y}\alpha_{x,y}})
\\\label{eq:short}&=e^{-\sum_{x,y}\frac{n}{s}(t_{x,y}P_{X,Y}(x,y)-{\epsilon}_{x,y})}
\prod_{j\in [1,c]}\mathbb{E}(e^{\frac{1}{s}\sum_{x,y}\sum_{k=1}^{i_j}t_{x,y}Z^{\{(x,y)\}}_{(\pi),i}}).
\end{align}
We need to investigate $\mathbb{E}(e^{\frac{1}{s}\sum_{x,y}\sum_{k=1}^{i_j}t_{x,y}Z^{\{(x,y)\}}_{(\pi),i}})$. 
Define $T_j^{(x,y)}=\sum_{k=1}^{i_j}Z^{\{(x,y)\}}_{(\pi),i}, j\in [1,c], x,y\in \mathcal{X}\times\mathcal{Y}$ as the number of occurrences of the pair $(x,y)$ in the $j$th cycle. Note that by definition, we have $\sum_{x,y}\sum_{k=1}^{i_j}Z^{\{(x,y)\}}_{(\pi),i}=\sum_{x,y}T^{(x,y)}_j=i_j$. Define $S^{(x,y)}_{j}=\frac{1}{s}T_j^{(x,y)}, j\in [1,c], x,y\in \mathcal{X}\times\mathcal{Y}$.
Let $\mathcal{B}=\{(s^{(x,y)}_j)_{j\in [1,c], x,y\in \mathcal{X}\times\mathcal{Y}}: \sum_{x,y}s^{(x,y)}_j=\frac{i_j}{s}, j \in [1,c]\}$ be the set of feasible values for the vector $(S^{(x,y)}_j)_{j\in [1,c], x,y\in \mathcal{X}\times\mathcal{Y}}$. We have: 
\begin{align*}
   \mathbb{E}(e^{\frac{1}{s}\sum_{x,y}\sum_{k=1}^{i_j}t_{x,y}Z^{\{(x,y)\}}_{(\pi),i}})
   &=\mathbb{E}(e^{\sum_{x,y}t_{x,y}\frac{1}{s}\sum_{k=1}^{i_j}Z^{\{(x,y)\}}_{(\pi),i}})
   =\mathbb{E}(e^{\sum_{x,y}t_{x,y}S^{(x,y)}_j})
   \\&=\sum_{(s^{\{(x,y)\}}_{j})_{j\in [1,c], x,y\in \mathcal{X}\times \mathcal{Y}}\in \beta}
P((s^{\{(x,y)\}}_{j})_{j\in [1,c], x,y\in \mathcal{X}\times \mathcal{Y}})e^{\sum_{x,y}t_{x,y}s^{\{(x,y)\}}_{j}}.
\end{align*}
For a fixed vector $(s^{\{(x,y)\}}_{j})_{j\in [1,c], x,y\in \mathcal{X}}\in \beta$, let $V^{(x,y)}$ be defined as the random variable for which $P(V^{(x,y)}=t_{(x,y)})=s^{\{(x,y)\}}_{j}, x,y\in \mathcal{X} $ and $P(V^{(x,y)}=0)=1-\frac{i_j}{s}$ (note that $P_V$ is a valid probability distribution). We have:
\begin{align*}
   &\mathbb{E}(e^{\frac{1}{s}\sum_{x,y}\sum_{k=1}^{i_j}t_{x,y}Z^{\{(x,y)\}}_{(\pi),i}})
   =\sum_{(s^{\{(x,y)\}}_{j})_{j\in [1,c], x,y\in \mathcal{X}\times \mathcal{Y}}\in \beta}
P((s^{\{(x,y)\}}_{j})_{j\in [1,c], x,y\in \mathcal{X}\times \mathcal{Y}})e^{\sum_{x,y}t_{x,y}s^{\{(x,y)\}}_{j}}
\\&=\sum_{(s^{\{(x,y)\}}_{j})_{j\in [1,c], x,y\in \mathcal{X}\times \mathcal{Y}}\in \beta}
P((s^{\{(x,y)\}}_{j})_{j\in [1,c], x,y\in \mathcal{X}\times \mathcal{Y}})e^{\mathbb{E}(V^{(x,y)})}\leq
\sum_{(s^{\{(x,y)\}}_{j})_{j\in [1,c], x,y\in \mathcal{X}\times \mathcal{Y}}\in \beta}
P((s^{\{(x,y)\}}_{j})_{j\in [1,c], x,y\in \mathcal{X}\times \mathcal{Y}})\mathbb{E}(e^{V^{(x,y)}}),
\end{align*}
where we have used Jensen's inequality in the last equation. Note that by construction, we have $\mathbb{E}(e^{V^{(x,y)}})= 1-\frac{i_j}{s}+\sum_{x,y}s^{(x,y)}_je^{t_{x,y}}$. Consequently:
\begin{align*}
   &\mathbb{E}(e^{\frac{1}{s}\sum_{x,y}\sum_{k=1}^{i_j}t_{x,y}Z^{\{(x,y)\}}_{(\pi),i}})
   \leq 
   \sum_{(s^{\{(x,y)\}}_{j})_{j\in [1,c], x,y\in \mathcal{X}\times \mathcal{Y}}\in \beta}
P((s^{\{(x,y)\}}_{j})_{j\in [1,c], x,y\in \mathcal{X}\times \mathcal{Y}})(1-\frac{i_j}{s}+\sum_{x,y}s^{(x,y)}_je^{t_{x,y}})
\\&= 
 1-\frac{i_j}{s}+\sum_{x,y}e^{t_{x,y}}\mathbb{E}(S_{j}^{(x,y)})
 =  1-\frac{i_j}{s}+\sum_{x,y}e^{t_{x,y}}\mathbb{E}(\frac{1}{s}\sum_{k=1}^{i_j}Z^{\{(x,y)\}}_{(\pi),i})
 \\& =  1-\frac{i_j}{s}+
 \frac{1}{s}\sum_{x,y}\sum_{k=1}^{i_j}e^{t_{x,y}}\mathbb{E}(Z^{\{(x,y)\}}_{(\pi),i})
 = 1-\frac{i_j}{s}+
 \frac{1}{s}\sum_{x,y}\sum_{k=1}^{i_j}e^{t_{x,y}}P_X(x)P_Y(y)
 \\& = 1-\frac{i_j}{s}+
 \frac{1}{s}\sum_{x,y}\sum_{k=1}^{i_j}P_{X,Y}(x,y)=1.
\end{align*}
Setting $\mathbb{E}(e^{\frac{1}{s}\sum_{x,y}\sum_{k=1}^{i_j}t_{x,y}Z^{\{(x,y)\}}_{(\pi),i}})\leq 1$ in Equation \eqref{eq:short}, we get: 
\begin{align*}
    P((X^n,\pi(Y^n))\in \mathcal{A}_{\epsilon}^n(X,Y))\leq e^{-\sum_{x,y}\frac{n}{s}(t_{x,y}P_{X,Y}(x,y)-{\epsilon}_{x,y})}= 2^{-\frac{n}{s}(I(X;Y)-{\epsilon}_{x,y})}.
\end{align*}
\section{Proof of Theorem \ref{th:2}}
\label{Ap:th:2}
The proof builds upon the arguments provided in the proof of Theorem \ref{th:1}. Let $Y^{n}=\pi_l(X_{(l)} ^n)_{l\in [k]}$.
First, we construct a partition $\mathsf{D}=\{\mathcal{C}_{j,t}:j\in [b_k], t\in [k(k-1)]\}$ such that each sequence of vectors $(Y_{(l), \mathcal{C}_{j,t}})_{l\in [k]}$ is an collection of independent vectors of i.i.d variables, where $Y_{(l), \mathcal{C}_{j,t}}=(Y_{(l),c})_{c\in  \mathcal{C}_{j,t}}$. Loosely speaking, this partitioning of the indices `breaks' the multi-letter correlation among the sequences induced due to the permutation and allows the application of standard information theoretic tools to bound the probability of joint typicality. The partition is constructed in two steps.  
We first construct a \textit{coarse} partition $\mathsf{C}=\{\mathcal{C}_1,\mathcal{C}_2,\cdots, \mathcal{C}_{b_k}\}$ of the indices $[1,n]$ for which the sequence of vectors $(Y_{(l),\mathcal{C}_j}), l\in [k]$ is identically distributed but not necessarily independent. The set $\mathcal{C}_j, j\in [b_k]$ is defined as the set of indices corresponding to partition $\mathcal{P}_j$, where correspondence is defined in Definition \ref{Def:corr}. Clearly, $\mathsf{C}=\{\mathcal{C}_1,\mathcal{C}_2,\cdots, \mathcal{C}_{b_k}\}$ partitions $[1,n]$ since each index corresponds to exactly one partition $\mathcal{P}_j$. To verify that the elements of the sequence $(Y_{(l),\mathcal{C}_j}), l\in [k]$ are identically distributed let us consider a fixed $j\in [b_k]$ and an arbitrary index $c\in \mathcal{C}_j$. Then the vector $(Y_{(1),c},Y_{(2),c},\cdots,Y_{(k),c})$ is distributed according to $P_{X_{\mathcal{P}_j}}$. To see this, note that:
\begin{align*}
    P_{Y_{(1),c},Y_{(2),c},\cdots,Y_{(k),c}}
    &= P_{X_{(1),(\pi_1^{-1}(c))},X_{(2),(\pi_2^{-1}(c))},\cdots,X_{(k),(\pi_k^{-1}(c))}}
\end{align*}
From the assumption that the index $c$ corresponds to the partition $\mathcal{P}_j$, we have that $\pi_l^{-1}(c)= \pi_{l'}^{-1}(c)$ if and only if $l,l'\in \mathcal{A}_{j,r}$ for some integer $r\in [|\mathcal{P}_j|]$. Since by the theorem statement $(X^n_{(l)})_{l\in [k]}$ is an i.i.d. sequence of vectors, the variables $X_{(l),\pi^{-1}_l(c)}$ and $X_{(l'),\pi^{-1}_{l'}(c)}$ are independent of each other if $\pi_l^{-1}(c)\neq \pi_{l'}^{-1}(c)$.  Consequently, 
\begin{align*}
    P_{Y_{(1),c},Y_{(2),c},\cdots,Y_{(k),c}}
    &= \prod_{r\in [|\mathcal{P}_j|]}P_{X_{t_1},X_{t_2},\cdots,X_{t_{|\mathcal{A}_{j,r}|}}}=P_{X_{\mathcal{P}_j}}.
\end{align*}
This proves that the sequences $(Y_{(l),\mathcal{C}_j}), l\in [k]$ are identically distributed with distribution $P_{X_{\mathcal{P}_j}}$.
In the next step, we decompose the partition $\mathsf{C}$ to arrive at a finer partition $\mathsf{D}=\{\mathcal{C}_{j,t}:j\in [b_k], t\in [k(k-1)]\}$ of $[1,n]$ such that $(Y_{(l),\mathcal{C}_{j,t}})_{l\in [k]}$ is an i.i.d sequence of vectors. Let $\mathcal{C}_j=\{c_1,c_2,\cdots,c_{|\mathcal{C}_j|}\}, j\in [b_k]$.
The previous step shows that the sequence consists of identically distributed vectors. In order to guarantee independence, we need to ensure that for any $c,c'\in \mathcal{C}_{j,t}$, we have $\pi^{-1}_{l}(c)
\neq \pi^{-1}_{l'}(c'), \forall l,l' \in [k]$. Then, independence of $(Y_{(l),c})_{l\in [k]}$ and $(Y_{(l),c'})_{l\in [k]}$ is guaranteed due to the independence of the sequence of vectors $(X^n_{(l)})_{l\in [k]}$. To this end we assign the indices in $\mathcal{C}_j$ to the sets $\mathcal{C}_{j,t}, t\in [k(k-1)]$ as follows:
\begin{align}
    &c_1\in \mathcal{C}_{j,1},\\
    &c_i\in \mathcal{C}_{j,l}: t= 
    \min \{t'| \nexists c'\in \mathcal{C}_{j,t'}, l,l' \in [k]: \pi^{-1}_l(c_i)=\pi^{-1}_{l'}(c')\}, i>1.
    \label{eq:assign}
\end{align}
Note that the set $\mathcal{C}_{j,t}$ defined in Equation \eqref{eq:assign} always exists since for any given $l \in [m]$, the value $\pi^{-1}_l(c)$ can be the same for at most $k$ distinct indices $c$ since each of the $k$ permutations maps one index to $\pi^{-1}_l(c)$. Furthermore, since $l$ takes $k$ distinct values, there are at most $k(k-1)-1$ indices $c'$ not equal to $c$ for which there exists $l,l'\in[k]$ such that $\pi_l(c)=\pi_{l'}(c')$. Since there are a total of $k(k-1)$ sets $\mathcal{C}_{j,t}$, by the Pigeonhole Principle, there exists at least one set for which there is no element $c'$ such that $\pi_l(c)=\pi_{l'}(c')$ for any value of $l,l'$. Consequently, $(Y_{(l),\mathcal{C}_{j,t}})_{l\in [k]}$ is an i.i.d. sequence with distribution $P_{X_{\mathcal{P}_{j}}}$.

Let $\underline{T}_{j,t}, j\in [b_k], t\in [k(k-1)]$ be the type of the sequence of vectors $(Y_{(l),\mathcal{C}_{j,t}})_{l\in [k]}$, so that $T_{j,t,x^k}= \frac{\sum_{c\in \mathcal{C}_{j,t}}\mathbbm{1}((Y_{(1),c},Y_{(2),c},\cdots, Y_{(k),c})=x^k)}{|\mathcal{C}_{j,t}|}, x^k\in \mathcal{X}^k$. We are interested in the probability of the event $(Y^n_{(l)})_{l\in [k]}\in \mathcal{A}_{\epsilon}^n(X^k)$. From Definition \ref{Def:ctyp} this event can be rewritten as follows:
\begin{align*}
 &P\left(\left(Y^n_{(l)})_{l\in [k]}\right)\in \mathcal{A}_{\epsilon}^n(X^k)\right)
  =P\left({T}((Y^n_{(l)})_{l\in [k]},x^m)\stackrel{.}{=} P_{X^k}(x^k)\pm \epsilon, \forall x^k\right)
 \\&= P(\sum_{j,t}\alpha_{j,t}{T}_{j,t,x^k}\stackrel{.}{=}P_{X^k}(x^k)\pm \epsilon, \forall x^m),
\end{align*}
where $\alpha_{j,t}= \frac{|\mathcal{C}_{j,t}|}{n}, j\in [b_k], t\in [k(k-1)]$, we write $a\stackrel{.}{=}x\pm \epsilon$ to denote $x-\epsilon \leq a\leq x_\epsilon$, and addition is defined element-wise. We have: 
\begin{align*}
 &P\left(\left(Y^n_{(l)})_{l\in [k]}\right)\in \mathcal{A}_{\epsilon}^n(X^k)\right)
 =\sum_{(\underline{s}^{b_k,k(k-1)})\in \mathcal{T}} P(\underline{T}_{j,t}=\underline{s}_{j,t},j\in [b_k], t\in [k(k-1)]),
\end{align*}
where $\mathcal{T}= \{(\underline{s}^{b_k,k(k-1)}:\sum_{j,t}\alpha_{j,t}{T}_{j,t,x^k}\stackrel{.}{=}P_{X^k}(x^k)\pm \epsilon, \forall x^k\}$. Using the property that for any set of events, the probability of the intersection is less than or equal to the geometric average of the individual probabilities, we have:
\begin{align*}
 &P((Y^n_{(l)})_{l\in [k]}\in \mathcal{A}_{\epsilon}^n(X^k))
\leq 
\sum_{(\underline{s}^{b_k,k(k-1)})\in \mathcal{T}} \sqrt[k(k-1)b_k]{\Pi_{i\in [j,t]}P(\underline{T}_{j,t}=\underline{s}_{j,t})}.
\end{align*}
Since the elements  $(Y_{(l),\mathcal{C}_{j,t}}), j\in [b_k], t\in [k(k-1)]$ are i.i.d by construction, it follows from standard information theoretic arguments \cite{csiszarbook} that:

\begin{align*}
    & P(\underline{T}_{j,t}=\underline{s}_{j,t}) \leq 2^{-|\mathcal{C}_{j,t}|(D(\underline{s}_i||P_{X_{\mathcal{P}_j}})-\prod_{l\in [k]}|\mathcal{X}_l|\epsilon)}, j\in [b_k], t\in [k(k-1)].
\end{align*}
We have, 
\begin{align*}
 &P((Y^n_{(l)})_{l\in [k]}\in \mathcal{A}_{\epsilon}^n(X^k))
 \leq  \sum_{(\underline{s}^{b_k,k(k-1)})\in \mathcal{T}} \sqrt[k(k-1)b_k]{\Pi_{i\in [j,t]} 2^{-|\mathcal{C}_{j,t}|(D(\underline{s}_i||P_{X_{\mathcal{P}_j}})-\prod_{l\in [k]}|\mathcal{X}_l|\epsilon)}}\\
 &\stackrel{(a)}{\leq}
\sum_{(\underline{s}^{b_k,k(k-1)})\in \mathcal{T}}
 \sqrt[k(k-1)b_k]
 {2^{-n(D(\sum_{j,t}\alpha_{j,t}\underline{s}_{j,t}
 ||\sum_{k}P_{X_{\mathcal{P}_j}})-\prod_{l\in [k]}|\mathcal{X}_l|\epsilon)}}
 \\&\stackrel{(b)}{\leq}  2^{-\frac{n}{k(k-1)b_k}(D(P_{X,Y}
 ||\sum_{j\in [b_k]}\frac{|\mathcal{C}_j|}{n}
 P_{X_{\mathcal{P}_j}})-\epsilon\prod_{l\in [k]}|\mathcal{X}_l|+O(\frac{\log{n}}{n}))}.
\end{align*}
where the (a) follows from the convexity of the divergence function and (b) follows by the fact that the number of joint types grows polynomially in $n$.
\section{Proof of Lemma \ref{lem:foldcount}}
\label{Ap:lem:foldcount}
The upper-bound follows by the fact that for $r$-fold derangement $(\pi_1(\cdot),\pi_2(\cdot),\cdots,\pi_k(\cdot))$, the first permutation is $\pi_1(\cdot)$ is the identity permutation, and the rest of derangements with respect to $\pi_1(\cdot)$, so by the counting principle there are at most $(!n)^{r-1}$ choices for $(\pi_1(\cdot),\pi_2(\cdot),\cdots,\pi_k(\cdot))$. Next we prove the lower bound. Note that $\pi_1(\cdot)$ is the identity permutation. By the same arguments as in the proof of Lemma \ref{lem:dercount}, there are at least $(n-1)!$ choices of distinct $\pi_2(\cdot)$, and for any fixed $\pi_2(\cdot)$ there are at least $(n-2)!$ distinct $\pi_3(\cdot)$. Generally, for fixed $\pi_2(\cdot),\pi_3(\cdot),\cdots, \pi_j(\cdot)$, there are at least $(n-j+1)!$ choices of distinct $\pi_{j+1}(\cdot)$. By the counting principle, there are at least $\prod_{j\in [r]}(n-j+1)!\geq ((n-r+1)!)^r$ distinct $(\pi_1(\cdot),\pi_2{\cdot},\cdots,\pi_{r}(\cdot))$. This completes the proof.

\section{Proof of Lemma \ref{lem:bell}}
\label{Ap:lem:bell}
First, we prove the upper-bound in Equation \eqref{eq:bell1}. As an initial step, we count the number of distinct allocations of partition correspondence to indices $i\in [1,n]$. Since we are considering $(i_1,i_2,\cdots, i_{b_k})$-Bell permutation vectors, there are a total of $i_j$ indices corresponding to $\mathcal{P}_j$ for $j\in [b_k]$. So, there are ${n \choose i_1,i_2,\cdots, i_{b_k}}$ allocations of partition correspondence to different indices. Now assume that the $i^{th}$ index corresponds to the $j$th partition. Then, we argue that there are at most $n^{|\mathcal{P}_j|}$ possible values for the vector $(\pi_{l}(i): l\in [k])$. The reason is that by definition, for any two $\pi_l(i)$ and $\pi_{l'}(i)$, their value are equal if and only if $l,l' \in \mathcal{A}_{j,r}$ for some integer $r\in [|\mathcal{P}_j|]$. So, the elements of $(\pi_{l}(i): l\in [m])$ take $|\mathcal{P}_j|$ distinct values among the set $[1,n]$. Consequently $(\pi_{l}(i): l\in [k])$ takes at most $n^{|\mathcal{P}_j|}$ distinct values. By the counting principle, the sequence of vectors $(\pi_{l}(i): l\in [k]), i\in [n]$ takes at most $n^{\sum_{j\in [b_k]} |\mathcal{P}_j|i_j-n}$  distinct values given a specific partition correspondence, since $\pi_1(\cdot)$ is assumed to be the identity permutation. Since there are a total of ${n \choose i_1,i_2,\cdots, i_{b_k}}$ partition correspondences, we have:   
\begin{align*}
    N_{i_1,i_2,\cdots,i_{b_k}} \leq {n \choose i_1,i_2,\cdots, i_{b_k}} n^{\sum_{j\in [b_k]} |\mathcal{P}_j|i_j-n}.
\end{align*}
Next, we prove the lower-bound in Equation \eqref{eq:bell1}. The proof follows by constructing enough distinct $(i_1,i_2,\cdots, i_{b_k})$-Bell permutation vectors. First, we choose a partition correspondence for the indices $i\in [n]$ similar to the proof for the lower-bound. There are ${n \choose i_1,i_2,\cdots, i_{b_k}}$ distinct ways of allocating the partition correspondence. We argue that for every fixed partition correspondence, there are at least $\prod_{j\in [b_k]]}d_{|\mathcal{P}_{j}
|}(i_j)$ permutations which are $(i_1,i_2,\cdots, i_{b_k})$-Bell permutation vectors. To see this,
without loss of generality, assume that the first $i_1$ indices $[1,i_1]$ correspond to $\mathcal{P}_1$, the next $i_2$ indices $[i_1+1,i_1+i_2]$ correspond to $\mathcal{P}_2$, and in general the indices $[\sum_{t=1}^{l-1}i_t+1, \sum_{t=1}^{l}i_t]$ correspond to $\mathcal{P}_j$. Let $(\pi'_{1,j},\pi'_{2,j},\cdots,\pi'_{|\mathcal{P}_j|,j})$ be vectors of $|\mathcal{P}_j|$-fold derangements of $[\sum_{t=1}^{j-1}i_t+1, \sum_{t=1}^{j}i_t]$, where $j\in [b_k]$. Then, the following is an $(i_1,i_2,\cdots,i_{b_k})$-Bell permutation vector. 
\begin{align*}
    \pi_l([\sum_{t=1}^{j-1}i_t+1, \sum_{t=1}^{j}i_t])= \pi'_{l,j}([\sum_{t=1}^{j-1}i_t+1, \sum_{t=1}^{j}i_t]), \quad \text{if } l\in \mathcal{A}_{s,j}, s\in [|\mathcal{P}_j|], j\in [b_k].
\end{align*}
There are a total of $d_{|\mathcal{P}_j|(i_j)}$ choices of $(\pi'_{1,j},\pi'_{2,j},\cdots,\pi'_{|\mathcal{P}_j|,j})$. So, by the counting principle, there are a total of $\prod_{j\in [b_k]}d_{|\mathcal{P}_j|}(i_j)$ choices of $(\pi_{1}(\cdot),\pi_2(\cdot),\cdots, \pi_{k}(\cdot))$ for a fixed partition correspondence. As argued previously, there are a total of  ${n \choose i_1,i_2,\cdots, i_{b_k}}$ distinct choices for partition correspondence. Consequently we have shown that, 
\begin{align*}
     {n \choose i_1,i_2,\cdots, i_{b_k}}\prod_{j\in [b_k]}d_{|\mathcal{P}_{j}|}(i_j)\leq N_{i_1,i_2,\cdots,i_{b_k}}.
\end{align*}
This completes the proof of Equation \eqref{eq:bell1}. We proceed with to prove Equation \eqref{eq:bell2}. Note that from the right hand side of Equation \eqref{eq:bell1}, we have:
\begin{align*}
 &\lim_{n\to\infty} \frac{\log_e{N_{i_1,i_2,\cdots,i_{b_k}}}}{n\log_e{n}}
 \leq
 \lim_{n\to\infty} \frac{\log_e{{n \choose i_1,i_2,\cdots, i_{b_k}} n^{(\sum_{j\in [b_k]} |\mathcal{P}_j|i_j-n)}}}{n\log_e{n}}
 = 
 \lim_{n\to\infty} \frac{\log_e{ n^{(\sum_{j\in [b_k]} |\mathcal{P}_j|i_j-n)}}}{n\log_e{n}}
 +\lim_{n\to\infty} \frac{\log_e{{n \choose i_1,i_2,\cdots, i_{b_k}}}}{n\log_e{n}}
 \\&
= \lim_{n\to\infty} \frac{{ ({\sum_{j\in [b_k]} |\mathcal{P}_j|i_j-n)}}}{n}
 +\lim_{n\to\infty} \frac{\log_e{2^n}}{n\log_e{n}}
 = {\sum_{j\in [b_k]} |\mathcal{P}_j|\alpha_j}-1.
\end{align*}
On the other hand, from the left hand side of Equation \eqref{eq:bell1}, we have:
\begin{align*}
    &\lim_{n\to\infty} \frac{\log{N_{i_1,i_2,\cdots,i_{b_k}}}}{n\log{n}}
    \geq \lim_{n\to\infty} \frac{\log{{n \choose i_1,i_2,\cdots, i_{b_k}}\prod_{j\in [b_k]}d_{|\mathcal{P}_{j}|}(i_j)}}{n\log{n}}
    \\&\stackrel{(a)}{\geq}
    \lim_{n\to\infty} \frac{\log{2^n\prod_{j\in [b_k]}d_{|\mathcal{P}_{j}|}(i_j)}}{n\log{n}}
    \stackrel{(b)}{\geq}
    \lim_{n\to\infty} \frac{\log{\prod_{j\in [b_k]}((i_j-|\mathcal{P}_j|+1)!^{|\mathcal{P}_j|-1})}}{n\log{n}}
    \\& = 
    \lim_{n\to\infty} \frac{\sum_{j\in [b_k]}{(|\mathcal{P}_j|-1)}\log{(i_j-|\mathcal{P}_j|+1)!}}{n\log{n}}
    \\&\stackrel{(c)}{=}
    \lim_{n\to\infty} \frac{\sum_{j\in [b_k]}{(|\mathcal{P}_j|-1)}({(i_j-|\mathcal{P}_j|+1)\log{(i_j-|\mathcal{P}_j|+1)}-(i_j-|\mathcal{P}_j|+1)+O(\log{(i_j-|\mathcal{P}_j|+1)}))}}{n\log{n}}
    \\&= \sum_{j\in [b_k]}|\mathcal{P}_j|\alpha_j-1,
\end{align*}
where (a) follows from the fact that ${n \choose i_1,i_2,\cdots, i_{b_k}}\leq 2^n$, (b) follows from Lemma \ref{lem:foldcount}, and in (c) we have used Stirling's approximation.



 \end{appendices}

\newpage
\bibliographystyle{IEEEtran}
\bibliography{ref}

\end{document}